%% file: main.tex
\documentclass{fmcad}

\include{preamble}

\begin{document}

\input{meta}
\maketitle
\pagestyle{plain}

\begin{abstract}
    We present a novel and efficient method for synthesis of parameterized
    distributed protocols by sketching. Our method is both syntax-guided and
    counterexample-guided, and utilizes a fast equivalence reduction technique
    that enables efficient completion of protocol sketches, often significantly
    reducing the search space of candidate completions by several orders of
    magnitude. To our knowledge, our tool, \tool, is the first synthesis tool
    for the widely used specification language TLA+. We evaluate \tool on a
    diverse benchmark of distributed protocols, demonstrating the ability to
    synthesize a large scale distributed Raft-based dynamic reconfiguration
    protocol beyond the scale of what existing synthesis techniques can handle.
\end{abstract}

\begin{IEEEkeywords}
    distributed protocols, synthesis, syntax-guided, counterexample-guided,
    sketching
\end{IEEEkeywords}

\input{sections/intro}

\input{sections/background}
\input{sections/problem}
\input{sections/algorithm}

\input{sections/encoding}

\input{sections/evaluation}

\input{sections/related_work}

\input{sections/conclusion}


\bibliographystyle{plain}
\bibliography{bib}

\appendix
\input{sections/theorems}

\end{document}

%% file: preamble.tex
\IEEEoverridecommandlockouts
\def\BibTeX{{\rm B\kern-.05em{\sc i\kern-.025em b}\kern-.08em
    T\kern-.1667em\lower.7ex\hbox{E}\kern-.125emX}}

\usepackage{amsmath,amssymb,amsfonts}
\usepackage{amsthm}
\usepackage{algorithmic}
\usepackage{graphicx}
\usepackage{textcomp}
\usepackage{xcolor}
\usepackage{mathtools}
\usepackage{xspace}
\usepackage{listings}
\usepackage[normalem]{ulem}
\newtheorem{problem}{Problem}
\newtheorem{theorem}{Theorem}

\DeclarePairedDelimiter\angles{\langle}{\rangle}

\newcommand{\Fig}{Fig.}
\newcommand{\Section}{Section}
\newcommand{\Table}{Table}

\newcommand{\tool}{\textsc{Scythe}\xspace}

\newcommand{\TLA}{TLA\textsuperscript{+}\xspace}
\newcommand{\Params}{\textsc{Params}\xspace}
\newcommand{\Vars}{\textsc{Vars}\xspace}
\newcommand{\Init}{\textsc{Init}\xspace}
\newcommand{\Next}{\textsc{Next}\xspace}
\newcommand{\Holes}{\textsc{Holes}\xspace}
\newcommand{\Const}{\textsc{Const}\xspace}

\newcommand{\prot}{\ensuremath{X}\xspace}
\newcommand{\run}{\ensuremath{r}\xspace}
\newcommand{\prop}{\ensuremath{\Phi}\xspace}

\newcommand{\raft}{raft-reconfig\xspace}
\newcommand{\spurious}{vacuous\xspace}

\newcommand{\True}{\textit{True}\xspace}
\newcommand{\False}{\textit{False}\xspace}
\newcommand{\Type}[1]{\textit{#1}\xspace}
\newcommand{\Bool}{\Type{Bool}}
\newcommand{\Int}{\Type{Int}}
\newcommand{\Set}[1]{\Type{(Set #1)}\xspace}
\newcommand{\Array}[2]{\Type{(Array #1 #2)}\xspace}
\newcommand{\OfDomain}[1]{\Type{(OfDomain #1)}\xspace}



%% file: meta.tex
\title{Efficient Synthesis of Symbolic Distributed Protocols by Sketching\\
%
\thanks{This material is partly supported by the National Science
Foundation under Graduate Research Fellowship Grant \#1938052,
and Award \#2319500. Any opinion,
findings, and conclusions or recommendations expressed in this material are
those of the authors(s) and do not necessarily reflect the views of the National
Science Foundation.}}

\author{\IEEEauthorblockN{Derek Egolf}
\IEEEauthorblockA{
    \textit{Northeastern University}\\
    Boston, MA \\
    egolf.d@northeastern.edu}
\and
\IEEEauthorblockN{William Schultz}
\IEEEauthorblockA{
    \textit{Northeastern University}\\
    Boston, MA \\
    schultz.w@northeastern.edu}
\and
\IEEEauthorblockN{Stavros Tripakis}
\IEEEauthorblockA{
    \textit{Northeastern University}\\
    Boston, MA \\
    stavros@northeastern.edu}
}

%% file: sections/intro.tex
\section{Introduction}
\label{sec_intro}

Distributed protocols have become a crucial component in the operation of modern
computer systems, including financial infrastructure
\cite{Buterin2013,Buchman2016TendermintBF} and cloud data storage systems
\cite{corbett2013spanner,decandia2007dynamo}. In addition to being consequential
and widely used, the complexity of these protocols makes them notoriously hard
to design and reason about.

Automated verification of distributed protocols has made great advances in
recent years. Specifically, inductive invariant inference methods have allowed
for fuller automation of the verification of safety
properties~\cite{2021ic3posymmetry,2021swisshance,yao2021distai,YaoTGN22,DBLP:conf/fmcad/SchultzDT22,schultz2024scalablearxiv}.
State of the art tools in this domain are able to verify non-trivial
specifications of {\em parameterized}, {\em infinite-state} protocols, written
in languages such as \TLA~\cite{lamport2002specifying} or
Ivy~\cite{2016padonivy}. Such verification efforts include not just protocol
specifications especially designed to fit into the decidable fragment of
Ivy~\cite{padonpaxosEPR}, but also generally undecidable specifications of
protocols such as Raft written in
\TLA~\cite{DBLP:conf/fmcad/SchultzDT22,schultz2024scalablearxiv} or Paxos
written in Ivy~\cite{padonpaxosEPR,GoelSakallahFMCAD2021}. Progress is also
being made towards fuller automation of the verification of liveness properties,
e.g., see~\cite{YaoLivenessPOPL2024}.

On the other hand, automated {\em synthesis} of distributed protocols is less
advanced. This discrepancy might be expected because synthesis is intuitively a
harder problem than verification: verification is about checking that a {\em
given} system is correct, while synthesis involves inventing a system {\em and}
ensuring that it is correct. Theory supports this intuition: model checking
finite-state distributed systems is decidable, but synthesis of finite-state
distributed systems is generally
undecidable~\cite{PnueliRosner90,Thistle2005,TripakisIPL}. Synthesis of
parameterized distributed protocols is also generally
undecidable~\cite{DBLP:journals/corr/JacobsB14}.

An easier problem than doing synthesis from scratch is to do synthesis by
sketching~\cite{LezamaAPLAS2009,ArmandoSTTT2013}. Sketching turns the synthesis
problem into a completion problem: given a {\em sketch} (i.e., an incomplete
system with {\em holes}) the goal is to complete the sketch such that the
completion satisfies a given correctness specification. The holes are typically
missing state variable updates, guards, or parts thereof. Completing a hole
means finding the missing expression. 

In this paper, we consider the problem of synthesis of distributed protocols by
sketching. Contrary to prior works that either apply only to special classes of
protocols~\cite{DBLP:journals/acta/MirzaieFJB20,DBLP:conf/tacas/JaberWJKS23}, or
target protocols in an ad-hoc specification
language~\cite{DBLP:conf/cav/AlurRSTU15}, our work targets general protocols
written in \TLA~\cite{lamport2002specifying}, a highly expressive specification
language with widespread use in both academia and the
industry~\cite{NewcombeAmazon2015}. 

Our approach follows the counterexample-guided inductive synthesis (CEGIS)
paradigm~\cite{ArmandoSTTT2013,GulwaniPolozovSingh2017}: a {\em learner} is
responsible for generating candidate solutions, while a {\em verifier} is
responsible for checking whether a candidate satisfies the requirements. 

Our synthesis method is truly {\em syntax-guided} in the sense that our
synthesis loop explores directly the space of candidate symbolic expressions
that can be generated from a given grammar. In contrast, previous
work~\cite{DBLP:conf/cav/AlurRSTU15} explores the space of (finite)
interpretations of uninterpreted functions. Our synthesis tool generates
expressions, whereas the synthesis tool of~\cite{DBLP:conf/cav/AlurRSTU15}
generates input-output tables (which can then be passed to an external SyGuS
solver~\cite{DBLP:conf/fmcad/AlurBJMRSSSTU13} as a post-processing step). Our
method does not rely on an external SyGuS solver. 

A crucial component of our synthesis algorithm is how exactly we generate
candidate expressions from a given grammar (or grammars, in the case of multiple
holes). A naive, breadth-first enumeration of all possible expressions in the
grammar does not scale. Instead, we use a novel technique that employs {\em
cache-based} enumeration coupled with an {\em equivalence reduction} with {\em
short-circuiting}. This technique allows us to not only avoid checking
semantically equivalent expressions, but also to avoid generating redundant
expressions in the first place. Indeed, some of our experiments show a reduction
in the number of generated expressions of more than three orders of magnitude.

We implement our method in a synthesis tool called \tool. \tool synthesizes
protocols that are parameterized, i.e., in a form that is directly generalizable
to an arbitrary number of nodes. Some of our synthesized protocols can be
instantiated with infinite-domain variables, i.e., we can handle infinite-state
protocols. \tool is able to synthesize complex expressions from non-trivial
grammars. For example, \tool can synthesize the guard expression
\begin{equation}
  \label{expressionintro}
\begin{aligned}
\forall Q_1 \in
\textit{Quorums}(\textit{config}[i]), \hspace{3cm} \\ \forall Q_2 \in
\textit{Quorums}
(\textit{new\_config}): Q_1 \cap Q_2 \neq \emptyset
\end{aligned}
\end{equation}
which is parameterized by $i$, the node that is executing reconfiguration, and it can also synthesize the state variable update
\begin{equation}
  \label{expressionintro2}
\begin{aligned}
\textit{votes}' = [\textit{votes}\ \text{EXCEPT}\ [n_1] = \textit{votes}[n_1] \cup \{n_2\}]
\end{aligned}
\end{equation}
which is parameterized by $n_1$ and $n_2$, the nodes that are exchanging votes.

We evaluate \tool on a suite of non-trivial benchmarks, including variants of
the Raft dynamic reconfiguration
protocol~\cite{ongaro2014search,SchultzDardikTripakisCPP2022,DBLP:conf/fmcad/SchultzDT22}.
\tool is able to synthesize correct and sometimes novel protocols in less than
an hour and often in a matter of minutes. Although \tool itself only guarantees
correctness for a finite protocol instance, we were able to prove a-posteriori
(using TLAPS~\cite{cousineau2012tlaps}) that the synthesized protocols are in
fact correct for an arbitrary number of nodes, as well as in some cases for
infinite-domain state variables.

In summary, this work makes the following contributions: (1) A novel distributed
protocol synthesis method that is both {\em syntax-guided} as well as {\em
counterexample-guided}. (2) Novel techniques to accelerate the search of
candidate completions, often reducing the search space by several orders of
magnitude. (3) The synthesis tool \tool which, to our knowledge, is the only
tool able to handle a diverse suite of real-world distributed protocol
benchmarks written in a broadly used language such as \TLA. (4) Formal
correctness proofs which demonstrate that \tool is able to synthesize
infinite-state, parameterized protocols that are safe for any protocol instance.

%% file: sections/background.tex
\section{Preliminaries}

\subsection{Protocol Representation in \TLA}

We consider symbolic transition systems modeled in
\TLA~\cite{lamport2002specifying}, e.g., as shown in
\Fig~\ref{fig:protocol_example}. A primed variable, e.g., $\textit{vote\_yes}'$,
denotes the value of the variable in the next state. Formally, a {\em protocol}
is a tuple $\angles{\Params, \Vars, \Init, \Next}$. \Params is a set of {\em
parameters} that may vary from one instantiation of the protocol to the other,
but do not change during the execution of the protocol (e.g. a set of node ids
\textit{Node} in \Fig~\ref{fig:protocol_example} line~\ref{line:constant}).
\Vars is the set of {\em state variables} (e.g. \Fig~\ref{fig:protocol_example}
line~\ref{line:vars}). \Init and \Next are predicates specifying, respectively,
the {\em initial states} and the {\em transition relation} of the system,  as
explained in detail in Sections~\ref{sec:protocol-semantics}
and~\ref{sec_conventions}.

\TLA is untyped, but for purposes of synthesis we assume that each symbol in
\Params and \Vars is typed. Supported  types include \Bool and \Int, and
sets or arrays of types. If \Type{T1} and \Type{T2} are types, then an
element of type \Set{T1} is a set of elements of type \Type{T1} and
an element of type \Array{T1}{T2} is a map from elements of type
\Type{T1} to elements of type \Type{T2}.

A tuple $\angles{\Const, \Vars, \Init, \Next}$ denotes an {\em instance} of a
protocol, where \Const is a mapping of \Params to values. For instance,
$[\textit{Node} \mapsto \{\textit{n1}, \textit{n2}, \textit{n3}\}]$
characterizes one instance of the protocol in \Fig~\ref{fig:protocol_example}
and $[\textit{Node} \mapsto \{\textit{a0},\textit{b0}\}]$ characterizes another.
The values of parameters need not be finite sets, e.g., \textit{MaxVal} might
have type \Int and specify a bound on some value. Note that a protocol is
technically not operational until the symbols in \Params are assigned to values,
since the valuation of \Init and \Next may depend on the valuation of the
symbols in \Params. 

A symbol in \Params may have type \textit{Domain}, which designates it as an
{\em opaque set}. If \textit{Prm} is a \Params symbol of type \textit{Domain}
and \Const assigns \textit{Prm} to set $P$, then an object $x$ has type
\OfDomain{Prm} if and only if $x\in P$. For instance, if symbol \textit{Node} in
\Fig~\ref{fig:protocol_example} has type \textit{Domain}, then the symbol
\textit{vote\_yes} has type \Type{(Set (OfDomain Node))}. (Alternatively,
\textit{Node} could have type \Type{(Set Int)}, but this typing would allow the
protocol to, e.g., do arithmetic on node ids, which may not be desirable.) We
discuss in \Section~\ref{sec:problem-statements} the use of opaque sets.

\subsection{Protocol Semantics}
\label{sec:protocol-semantics}

A {\em state} of a protocol instance is an assignment of values to the variables
in \Vars. \Init is a predicate mapping a state to true or false; if a state
satisfies \Init (if it maps to true), it is an initial state of the protocol.

The transition relation \Next is a predicate mapping a pair of states to true or
false. If a pair of states $(s,t)$ satisfies \Next, then there is a {\em
transition} from $s$ to $t$, and we write $s \rightarrow t$. A state is {\em
reachable} if there exists a {\em run} of the protocol instance containing that
state. A run of a protocol instance is a possibly infinite sequence of states
$s_0, s_1, s_2 ...$ such that (1) $s_0$ satifies \Init, (2) $s_i \rightarrow
s_{i+1}$ for all $i \geq 0$, and (3) the sequence satisfies optional {\em
fairness constraints}.  We omit a detailed discussion of fairness. At a
high-level, some transitions are called {\em fair} and under certain conditions,
they must be taken. In this way, certain sequences of states are excluded from
the set of runs of the protocol. In particular, if a sequence of states that
would otherwise be a run ends in a certain cycle, that sequence may be excluded
from the set of runs of a protocol due to fairness constraints.

\subsection{Properties and Verification}

We support standard temporal safety and liveness {\em properties} for specifying
protocol correctness. Safety is often specified using a state {\em invariant}: a
predicate mapping a state to true or false. A protocol instance satisfies a
state invariant if all reachable states satisfy the invariant. A protocol
instance satisfies a temporal property if all runs (or fair runs, if fairness is
assumed) satisfy the property. A protocol satisfies a property if all its
protocol instances satisfy the property.

\input{figs/protocol_example.tex}

\subsection{Modeling Conventions}
\label{sec_conventions}

We adopt  
standard conventions on the syntax used to represent protocols, particularly on
how \Next is written. Specifically, we decompose \Next into a disjunction of
{\em actions} (e.g. \Fig~\ref{fig:protocol_example}
lines~\ref{line:next}-\ref{line:next-vote-yes}). An action is a predicate
mapping a pair of states to true or false; e.g., action \textit{GoCommit} of
\Fig~\ref{fig:protocol_example}. We  decompose an action into the conjunction of
a {\em pre-condition} and a {\em post-condition}. A pre-condition is a predicate
mapping a state to true or false; if the pre-condition of an action is satisfied
by a state, then we say the action is {\em enabled} at that state. For instance,
\Fig~\ref{fig:protocol_example} line~\ref{line:go-commit-vote-yes} says that
action \textit{GoCommit} is enabled only when all nodes have voted yes.

We decompose a post-condition into a conjunction of {\em post-clauses}, one for
each state variable. A post-clause determines how its associated state variable
changes when the action is taken. For instance, \Fig~\ref{fig:protocol_example}
line~\ref{line:go-commit-go-commit} shows a post-clause for the state variable
\textit{go\_commit}, denoted by priming the variable name:
$\textit{go\_commit}'$.

In general, post-clauses may be arbitrary predicates involving the primed
variable (e.g. $v' \in e$). We assume that all synthesized post-clauses are of
the form $v' = e$ where $e$ does not contain any primed variables, but make no
assumptions on the post-clauses that we do not synthesize. In synthesis,
non-determism is used extensively, e.g., for modeling the environment. We note
that the $v' = e$ assumption does not limit us to deterministic protocols, since
multiple actions may be enabled at the same state. 

Some actions are {\em parameterized}. For instance on line \ref{line:vote-yes}
of Figure \ref{fig:protocol_example}, the action \textit{VoteYes} is
parameterized by a symbol $n$. From line \ref{line:next-vote-yes}, we can infer
that $n$ denotes an element of the set \textit{Node}. The {\em arguments} of an
action are those symbols like $n$. The {\em domain} of an argument to an action
is the set quantified over for that argument in \Next. For example, the domain
of argument $n$ of action \textit{VoteYes} is the set \textit{Node}. We require
that the domain of an argument be a symbol from \Params of type \textit{Domain}.
An action may have multiple arguments; the domain of an action is the Cartesian
product of the domains of its arguments. A parameterized action denotes a family
of actions, one for each element in its domain. 

If $s \rightarrow t$ is a transition and $(s,t)$ satisfies an action $A$, we can
say that {\em $A$ is taken} and write $s \xrightarrow{A} t$. Note that $(s,t)$
may satisfy multiple actions and we may annotate the transition with any of
them. We write $s \xrightarrow{A(\vec{v})} t$ to explicitly denote the arguments
to $A$; $\vec{v}$ is empty in the case of non-parameterized actions. In this
way, runs of a protocol may be outfitted with a sequence of actions. Annotating
runs of a protocol with actions is critical for our synthesis algorithm, since
annotations allow us to ``blame'' particular actions for causing a
counterexample run (c.f. \Section~\ref{sec:counterexample_generalization}).
Fairness constraints are often specified using actions: we may say {\em $A$ is
(strongly) fair} to mean that action $A$ must be taken if it is enabled
infinitely often. 

%% file: figs/protocol_example.tex
\begin{figure}
\begin{lstlisting}[mathescape,numbers=left,xleftmargin=20pt,escapechar=|]
$\text{CONSTANT Node}$|\label{line:constant}|
$\textit{vars} := (\textit{vote\_yes},\textit{go\_commit},\textit{go\_abort})$|\label{line:vars}|
$\textit{GoCommit} :=$|\label{line:go-commit}|
    $\wedge\ \textit{vote\_yes} = \textit{Node}$|\label{line:go-commit-vote-yes}|
    $\wedge\ \textit{go\_commit}' = \textit{Node}$|\label{line:go-commit-go-commit}|
    $\wedge\ \textit{go\_abort}' = \textit{go\_abort}$|\label{line:go-commit-go-abort}|
$\textit{VoteYes}(n) :=$|\label{line:vote-yes}|
    $\wedge\ \textit{vote\_yes}' = \textit{vote\_yes} \cup \{n\}$|$\label{line:vote-yes-vote-yes}$|
    $\wedge\ \textit{go\_commit}' = \textit{go\_commit}$
    $\wedge\ \textit{go\_abort}' = \textit{go\_abort}$
$\Init :=$|\label{line:init}|
    $\wedge\ \textit{vote\_yes} = \emptyset$
    $\wedge\ \textit{go\_commit} = \emptyset$ 
    $\wedge\ \textit{go\_abort} = \emptyset$
$\Next :=$|\label{line:next}|
    $\vee\ \textit{GoCommit}$|\label{line:next-go1}|
    $\vee\ \exists n \in \textit{Node} : \textit{VoteYes}(n)$|\label{line:next-vote-yes}|
\end{lstlisting}
\vspace{-1em}
\caption{An example of a \TLA protocol (excerpt).}
\vspace{-1.5em}
\label{fig:protocol_example}
\end{figure}

%% file: sections/problem.tex
\section{Synthesis of Distributed Protocols}

\subsection{Protocol Sketches}

A tuple $\angles{\Params,\Vars,\Holes,\Init,\Next_0}$ is a {\em protocol
sketch}, where \Params, \Vars, and \Init are as in a \TLA protocol and $\Next_0$
is a transition relation predicate containing the hole names found in \Holes.
\Holes is a (possibly empty) set of tuples, each containing a hole name $h$, a
list of argument symbols $\vec v_h$, an output type $t_h$, and a grammar $G_h$.
A hole represents an uninterpreted function of type $t_h$ over the arguments
$\vec v_h$. Each hole is associated with exactly one action $A_h$ and it appears
exactly once in that action. The grammar of a hole defines the set of candidate
expressions that can fill the hole. 

For example, a sketch can be derived from \Fig~\ref{fig:protocol_example} by
replacing the update of line~\ref{line:vote-yes-vote-yes} with
$\textit{vote\_yes}' = h(\textit{vote\_yes}, n)$, where $h$ is the hole name,
the hole has arguments \textit{vote\_yes} and $n$, the return type is \Type{(Set
(OfDomain Node))}, and the action of the hole is $\textit{VoteYes}(n)$. One
grammar for this hole might be (in Backus Normal Form): 
$$E ::= \emptyset\ |\ \{n\}\ |\ \textit{vote\_yes}\ |\ (E\cup E)\ |\ (E\cap E)\ |\ (E\setminus E) $$
which generates all standard set expressions over the empty set, the singleton
set $\{n\}$, and the set \textit{vote\_yes}. We note that, in general, each hole
of a sketch may have its own distinct grammar. 

A hole is either a {\em pre-hole} or a {\em post-hole}. If the hole is a
pre-hole, it is a placeholder for a pre-condition of the action. If the hole is
a post-hole, it is a placeholder for the right-hand side of a post-clause of the
action, e.g., as in $\textit{vote\_yes}' = h(\textit{vote\_yes}, n)$, where
$h$ is a post-hole. We do not consider synthesis of the initial state
predicate and therefore no holes appear in \Init.

The arguments of a hole $h$ may include any of the protocol parameters in
\Params, the state variables in \Vars, and the arguments of $A_h$ if the action
is parameterized. If $h$ is a pre-hole, then its return type is boolean. If the
hole is a post-hole, its type is the same as its associated variable, e.g., hole
$h$ above has the same type as \textit{vote\_yes}.

\subsection{Problem Statements}
\label{sec:problem-statements}

A {\em completion} of a sketch is a protocol derived from the sketch by
replacing each hole with an expression from its grammar. Informally, the
synthesis task is to find a completion of the protocol that satisfies a given
property. The distinction between a protocol and an instance of a protocol is
important here; it may be easier to find a completion of a protocol such that a
specific (e.g., finite) instance satisfies a property than to find a completion
such that all instances satisfy the property. Therefore, we define two versions
of the synthesis problem:

\begin{problem}
\label{problem:instance}
Let $\angles{\Params,\Vars,\Holes,\Init,\Next_0}$ be a sketch and \prop a
property. Let \Const be an assignment to \Params. Find a completion,
$\angles{\Params,\Vars,\Init,\Next}$, of the sketch such that the instance
$\angles{\Const,\Vars,\Init,\Next}$ satisfies \prop.
\end{problem}

\begin{problem}
\label{problem:protocol}
Let $\angles{\Params,\Vars,\Holes,\Init,\Next_0}$ be a sketch and \prop a
property. Find a completion, $\angles{\Params,\Vars,\Init,\Next}$, of the sketch
such that every instance of the completion satisfies \prop.
\end{problem}

In this paper we focus on solving Problem~\ref{problem:instance}. It is a
more tractable problem and we are able to use a model checker as a subroutine in
cases where the instance has finitely many states. It turns out that in many
cases, a solution to Problem~\ref{problem:instance} generalizes and is also a
solution to Problem~\ref{problem:protocol}. This generalizability comes from the
fact that the symbols in \Params are opaque; e.g., we may refer to the set of
nodes \textit{Node}, but we cannot refer to any particular element of
\textit{Node} without quantification. Indeed, as we show in
\Section~\ref{sec:evaluation}, our tool is able to synthesize protocols that
generalize, i.e., they are also solutions to Problem~\ref{problem:protocol}.

%% file: sections/algorithm.tex
\section{Our Approach}
\label{sec:algorithm}

As mentioned in the introduction, we follow the CEGIS paradigm which includes
two main components: a learner and a verifier. In our case, the learner and
verifier interact in a loop with the following steps:
(1) the learner generates a candidate completion \prot, if one exists, that
satisfies a (possibly empty) set of {\em pruning constraints} (i.e., \prot is
pruned if it {\em violates} the constraints),
(2) the verifier checks \prot against the supplied property \prop, 
(3) if \prot satisfies \prop, a solution is found and the algorithm terminates,
(4) if \prot does not satisfy \prop, the verifier produces a counterexample run
$\run$,
(5) the learner uses $\run$ to add new pruning constraints, and we repeat until
a solution is found or the search space is exhausted. 

Our learner component has three subcomponents: the {\em expression generator}
(EG), the {\em pruning constraint checker} (PCC), and the {\em counterexample
generalizer} (CXG). EG generates expressions from grammars, as detailed in
\Section~\ref{sec_EG}. PCC checks each generated expression against the current
set of pruning constraints, as explained in
\Section~\ref{sec:pruning_constraints}. CXG is invoked in Step~(5) to update the
pruning constraints by {\em generalizing} the information contained in the
counterexamples, as detailed in
\Section~\ref{sec:counterexample_generalization}. 

Pruning constraints eliminate candidate completions that are guaranteed to
exhibit previously encountered counterexamples, without having these candidates
checked by the verifier, which is often an expensive subroutine. A naive way to
do that would be to keep a list $L$ of counterexamples seen so far, and then
check whether a candidate exhibits any of the runs in $L$. Instead, we use more
sophisticated pruning constraints that encode counterexamples as logical
constraints on uninterpreted functions, c.f.
Sections~\ref{sec:pruning_constraints}
and~\ref{sec:counterexample_generalization}.

As our verifier in Step~(2), we use an off-the-shelf \TLA model checker,
specifically TLC~\cite{tlcmodelchecker}. We will not discuss \TLA model checking
further as it is standard.

\subsection{Expression Generation}
\label{sec_EG}

Recall that candidate protocols are completions of some sketch, which are in
turn charaterized as members of some grammar. In the case of multi-hole
sketches, completions are characterized as members of the cross-product of the
grammars. Therefore, generating candidate protocols reduces to {\em enumerating}
expressions from grammars. Note that, although grammars have a finite
representation, the {\em language} (i.e., set of expressions) of a grammar may
be infinite and the expressions therein may be arbitrarily large. 

We experimented with three grammar enumeration techniques: (1) a naive
breadth-first algorithm, (2) a cache-based algorithm, and (3) an extension of
the cache-based algorithm that exploits semantic equivalence of expressions. 

\subsubsection{Naive Breadth-First Algorithm}
A naive breadth-first search algorithm is to keep a priority queue (sorted by
size) of {\em partial expressions}, i.e. expressions containing both terminals
and non-terminals. A partial expression is discharged from the queue by
considering all possible ways to replace the non-terminals using the grammar
rules and substituting those into the partial expression. For example, if the
partial expression is $d := E \cup (x \cup E)$ and the grammar has a production
$G := E ::= x\ |\ y\ |\ E \cup E$, we would consider nine different partial
expressions, one for each pair of productions of the $E$ rule, since $E$ appears
twice in the expression $d$. After substituting, we can immediately return those
expressions which do not contain non-terminals and add to the queue those that
do. Our experience was that this algorithm was too slow in practice, since it
iterates over and performs substitutions on larger and larger partial
expressions.

\subsubsection{Cache-Based Algorithm}
In the cache-based algorithm, our learner generates all candidates of size $n$
before it generates any candidates of size $n+1$. Expressions are essentially
trees and our notion of {\em size} is the number of nodes in the tree, e.g., the
size of $(a + b) + c$ is 5. We keep a cache mapping each integer $n$ to the set
of all non-partial expressions of that size, for each non-terminal. We then use
this cache to build larger non-partial expressions, substituting only into
productions (partial expressions) that appear in the grammar. There are often
many expressions of size $n$. We use generators to yield a stream of
expressions, which avoids generating all expressions of a given size at once.

As an example of the cache-based algorithm, suppose we want to generate the
expressions of size 5 for the non-terminal $E$ in the grammar $G$ above. Assume
we already have a cache containing all expressions of size 1,2,3, and 4 for $E$.
Then we can generate all expressions of size 5 by substituting pairs of
expressions into the rule $E ::= E \cup E$ such that the sum of the sizes of the
two expressions is 5. 

\subsubsection{Equivalence Reduction}
Because the cache-based algorithm reuses all expressions of a given size many
times over, it is important to keep the cache as small as possible. In
particular, if two expressions are {\em semantically equivalent}, only one
should appear in the cache. To illustrate, consider that there are only 16
boolean expressions over two variables, modulo equivalence. The grammar $B ::=
x\ |\ y\ |\ \neg B\ |\ B \wedge B$ can express all 16 of these expressions, but
it generates infinitely many expressions. The number of expressions of size $n$
is $O(2^n)$. 

When we generate a new expression, we compute a normal form for that expression.
We then check if we have already generated an expression with that normal form.
If we have, we do not return the new expression and we do not add it to the
cache. We implement normal forms for (1) set expressions containing the
operations $\cup$, $\cap$, and $\setminus$, (2) boolean expressions containing
the operations $\vee$, $\wedge$, and $\neg$, (3) equality expressions, and (4)
inequality expressions. We use DNF as the normal form for boolean expressions.
Our normal form for set expressions exploits the correspondence between set
expressions and boolean functions and then uses DNF. Equality between sets $A$
and $B$ is equivalent to $\emptyset = (A \setminus B) \cup (B \setminus A)$; we
exploit this fact to obtain a normal form for equality between two sets. 

In addition to equivalence reduction by normal forms, we also exploit the
semantics of expressions to {\em short-circuit} the generation of expressions.
Short-circuiting is a technique that allows us to avoid iterating over large
parts of the search space. For instance, if we are generating expressions of
size 5 for the rule $E\cup E$, we can consider pairs of expressions of sizes
(1,4) and (2,3), but we can exploit the commutativity of union by ignoring sizes
(4,1) and (3,2). Without this technique, we would have to iterate over twice as
many pairs of expressions, compute their normal forms, and check if these normal
forms are in the cache. In general, for commutative operation $\odot$ if we are
generating expression $e_1\odot e_2$, we first pick $e_1$ and only iterate over
choices for $e_2$ that are at least as large as $e_1$.

%% file: sections/encoding.tex
\subsection{Counterexamples and Pruning Constraints}
\label{sec:pruning_constraints}

\subsubsection{Counterexamples}

A counterexample is a run of the protocol annotated with actions: $s_0
\xrightarrow{A_1(\vec{v_1})} s_1 \xrightarrow{A_2(\vec{v_2})} \ldots
\xrightarrow{A_k(\vec{v_k})} s_k$. The run is reported as a safety, deadlock, or
liveness violation. If the run is a safety or deadlock violation, it is
interpreted as a path. If it is a liveness violation, it is interpreted as a
path leading to a cycle, called a {\em lasso}. In the case of liveness, $s_k =
s_i$ for some $i < k$ and $s_i$ is the state that first injects into the cycle.

\subsubsection{Pruning Constraints}
Our pruning constraints are logical constraints over propositional logic with
equality and uninterpreted functions. For example, suppose we have the holes
$h_1(a,b)$ and $h_2(b,c)$. Then, an example of a pruning constraint is the
formula $\pi := (h_1(0, 1) \neq \True) \lor (h_2(1, 2) \neq 1)$. $\pi$
constrains the candidate expressions for the holes $h_1$ and $h_2$. For example,
replacing $h_1$ and $h_2$ with the expressions $a<b$ and $c-b$, respectively,
violates $\pi$, because $0<1=\True$ and $2-1=1$. Replacing $h_1$ and $h_2$ with
$a<b$ and $b$, respectively, also violates $\pi$. Hence, $\pi$ prunes at least
two completions.

Formally, a pruning constraint is a disjunction of {\em terms}, where each term
is a triple containing (1) a hole $h$, (2) a mapping $s^\star$ from the
arguments of $h$ to values, and (3) a literal value of the output type of $h$.
For instance, in $\pi$ above, the first term has $h = h_1(a,b)$, $s^\star =
[a\mapsto 0, b\mapsto 1]$, and $y = \True$. The second term has $h = h_2(b,c)$,
$s^\star = [b\mapsto 1, c\mapsto 2]$, and $y = 1$. Let $\tau := (h, s^\star, y)$
be a term and let $\widehat h$ be an interpretation (in our case, an expression)
for the uninterpreted function $h$. Then $\widehat h$ satisfies $\tau$ if
$\widehat h(s^\star) \neq y$. If $h_1, h_2,...., h_m$ are the uninterpreted
functions in a pruning constraint $\pi$ and $X := \widehat h_1, \widehat h_2,
..., \widehat h_m$ are interpretations for the $h_i$ (i.e. a completion), then
$X$ satisfies $\pi$ if the disjunction of the $\tau$ terms in $\pi$ is
satisfied.

In each run, our algorithm maintains a {\em set} of pruning constraints,
interpreted as a {\em conjunction} (of disjunctions of terms). A completion
satisfies a set of pruning constraints if it satisfies all constraints in the
set. Because we want to avoid seeing any counterexample more than once, the
learner will pass a completion $X$ to the verifier only if $X$ satisfies every
pruning constraint. I.e., a pruning constraint $\pi$ prunes completions that do
{\em not} satisfy $\pi$. PCC checks against the pruning constraints by
substituting the expressions of the holes into the constraints and performing
evaluation. Each type of counterexample (safety, deadlock, liveness) requires a
slightly different encoding as a pruning constraint, as explained next.

\subsection{Counterexample Generalization}
\label{sec:counterexample_generalization}

A pruning constraint $\pi$ is {\em under-pruning} w.r.t. run $\run$ and sketch
$S$ if there exists a completion $\prot$ of $S$ such that $\prot$ satisfies
$\pi$ and $\run$ is a run of $\prot$. $\pi$ is {\em over-pruning} w.r.t. run
$\run$ and sketch $S$ if there exists a completion $\prot$ of $S$ such that
$\prot$ does not satisfy $\pi$ and $\run$ is not a run of $\prot$. $\pi$ is {\em
optimal} if it is neither under- nor over-pruning. $\pi$ is {\em sub-optimal} if
it is under-pruning, but not over-pruning. Our primary goal is to avoid
over-pruning constraints, since over-pruning results in an incomplete algorithm,
i.e., an algorithm that might miss valid completions.

In what follows we present three techniques to encode into pruning constraints,
safety, deadlock, and liveness counterexamples, respectively. Our safety pruning
constraints are optimal (Theorem~\ref{thm:safe}), but our deadlock and liveness
pruning constraints are sub-optimal (Theorems~\ref{thm:dead}
and~\ref{thm:live}). In practice, these sub-optimal constraints are sufficient
to avoid many completions that exhibit the corresponding violations; the
bottleneck in our experiments is not the number of model checker calls. 

\subsubsection{Encoding Safety Counterexamples}

Intuitively, a safety violation can be fixed by ``cutting'' at least one
transition in the counterexample run, either by violating its guard or by
modifying its state update. Let $r = s_0 \xrightarrow{A_1(\vec{v_1})} s_1
\xrightarrow{A_2(\vec{v_2})} \ldots \xrightarrow{A_k(\vec{v_k})} s_k$ be a
safety violation and suppose that the completion is characterized by the
interpretations $\widehat h_1, \widehat h_2,... \widehat h_m$. We denote the
pruning constraint for \run as $\pi_\textit{safe}(\run)$ and construct it as
follows. $\pi_\textit{safe}(\run)$ is a disjunction of {\em $\tau$-terms}. For
each $s \xrightarrow{A(\vec v)} t$ in the counterexample, we construct a set of
$\tau$-terms. In particular, for each hole $h_i$ in the action $A$, we construct
the term $\tau_{A(\vec v),i} := (h_i, s^\star, y)$, where $y := \widehat
h_i(s^*)$ and where $s^\star$ is the predecessor state $s$, restricted to the
arguments of $h_i$, including the arguments to the action $A$. The pruning
constraint is then the disjunction containing all $\tau_{A(\vec v),i}$.

For instance, suppose the safety violation is $[a,b,c\mapsto 0,1,2]
\xrightarrow{A} [a,b,c\mapsto 1,1,2]$. Suppose additionally that $h_1(a,b)$ is a
pre-hole in $A$ and $a' = h_2(b,c)$ is a post-hole in $A$. Suppose that the
completion that resulted in the safety violation had $\widehat h_1(0,1) = \True$
and $\widehat h_2(1,2) = 1$. Then $\tau_{A,1} = (h_1, [a\mapsto 0, b\mapsto 1],
\True)$ and $\tau_{A,2} = (h_2, [b\mapsto 1, c\mapsto 2], 1)$. The pruning
constraint would be $\tau_{A,1}\vee\tau_{A,2}$, which corresponds to $\pi$ from
before. This constraint ensures that the pre-condition of $A$ is not satisfied
in the state $[a,b,c\mapsto 0,1,2]$ or that $a \neq 1$ after taking action $A$
in that state.

\subsubsection{Encoding Deadlock Counterexamples}

Informally, a pruning constraint of a deadlock violation is similar to that of a
safety violation because a deadlock violation can be fixed by making the
deadlocked state $s_k$ unreachable. But another way to fix a deadlock violation
is to make $s_k$ {\em undeadlocked}, which may be done by weakening the
pre-condition of some action that is not enabled in $s_k$.

Formally, the deadlock pruning constraint for run \run is defined to be
$\pi_\textit{dead}(\run) := \pi_\textit{safe}(\run) \lor \pi_\rho(\run)$, where
$\pi_\rho(\run)$ is a disjunction of {\em $\rho$-terms}, each of the form
$\rho_{A(\vec v),i, k} := (h_i, s_k^\star, y)$, where $s_k^\star$ is $s_k$
restricted to the arguments of $h_i$ and where $y := \widehat h_i(s_k^\star)$.
We construct a $\rho$-term for every action $A$ and every pre-hole $h_i$ in $A$
such that $\widehat h_i(s_k) = \False$. Then $\pi_\rho(\run)$ is the disjunction
of all all $\rho$-terms.

\subsubsection{Encoding Liveness Counterexamples}

The constraint for a liveness violation can be thought of as a generalization of
the constraint for a deadlock violation. It is sufficient to do one of (1) break
the path to the cycle using $\tau$-terms, (2) break the cycle using
$\tau$-terms, or (3) weaken the pre-condition of some fair action that is not
enabled in some state of the cycle using $\rho$-terms, making the cycle {\em
unfair}. Formally, we denote our liveness pruning constraint as
$\pi_\textit{live}(\run)$. We construct it as $\pi_\textit{live}(\run) :=
\pi_\textit{safe}(\run) \lor \pi_\rho'$, where $\pi_\rho'$ is the disjunction of
the following $\rho$-terms: For each fair action $A$, for every $\vec v$ in the
domain of $A$, for every $j$ such that $s_j$ is in the cycle, and for every {\em
pre-hole} $h_i$ in $A$ such that $\widehat h_i(s_j) = \False$, we construct the
term $\rho_{A(\vec v),i, j}$.

\subsubsection{Fairness and Stuttering}
Although we are able to handle both weakly and strongly fair actions, we did not
treat them differently above in $\pi_\textit{live}$. That construction may
be under-pruning in the presence of weakly fair actions, but it will
never over-prune and therefore our algorithm is complete. None of our benchmarks
required weak fairness when modeling the synthesized protocols. 

{\em Stuttering} (a special liveness violation) occurs when there are no fair,
enabled, non-self-looping actions in the final state of the violation. In
constrast, deadlock violations occur when there is no enabled action at all. We
denote the pruning constraint for a stuttering violation as
$\pi_\textit{stut}(\run) := \pi_\textit{safe}(\run) \lor \pi_\tau \lor
\pi_\rho'$. In addition to the $\tau$-terms from $\pi_\textit{safe}(\run)$, we
add $\pi_\tau$, which is the disjunction of $\tau_{A(\vec v), i}$ for every
post-hole $h_i$ in every fair action $A$. We add $\pi_\rho'$ as we did for a
typical liveness violation, except the only $s_j$ in the cycle is the last state
of \run, $s_k$.

\begin{theorem}
\label{thm_all}
    Let \run be a counterexample of a completion of the sketch $S$. If \run is a
    safety violation then $\pi_\textit{safe}(\run)$ is optimal w.r.t. \run and
    $S$. If \run is a deadlock, liveness, or stuttering violation then
    $\pi_\textit{dead}(\run)$, $\pi_\textit{live}(\run)$, and
    $\pi_\textit{stut}(\run)$, respectively, are sub-optimal w.r.t. \run and
    $S$. {\em--- The proof can be found in
    Appendix~\ref{sec:proof-theorem-all}.}
\end{theorem}

%% file: sections/evaluation.tex
\section{Implementation and Evaluation}
\label{sec:evaluation}

\subsubsection*{Implementation and Experimental Setup}
We implemented our method (\Section~\ref{sec:algorithm}) in a tool, \tool, which
supports many features of the \TLA language and utilizes the TLC model
checker~\cite{tlcmodelchecker} as verifier. \tool is written in Python and takes
as input (1) a TLA+ file defining the protocol and its sketch and (2) a
configuration file defining the grammars and types along with protocol
parameters. Our grammars are typed regular tree
grammars~\cite{DBLP:journals/corr/Engelfriet15} and our implementation
essentially uses the standard SYNTH-LIB input format for
SyGuS~\cite{DBLP:conf/fmcad/AlurBJMRSSSTU13}. We ran each experiment on a
dedicated 2.40 GHz CPU.

\subsubsection*{Benchmarks}
Our benchmark suite contains seven distinct protocols: (1) decentralized lock
service (decentr. lock), (2) server-client lock service (lock\_serv), (3)
Peterson's algorithm for mutual exclusion, (4) two phase commit (2PC), (5)
consensus, (6) sharded key-value store (sharded\_kv), (7) \raft, and (8)
\raft-big. (7) and (8) are non-trivial, reconfigurable variants of the Raft
protocol~\cite{ongaro2014search,SchultzDardikTripakisCPP2022,DBLP:conf/fmcad/SchultzDT22}.
Our benchmarks are adapted from safety verification benchmarks that have been
used in recent years \cite{DBLP:conf/fmcad/SchultzDT22,ivybench}. These existing
benchmarks contain a suite of correct, manually crafted protocols and we refer
to each manually crafted solution as the {\em ground truth}. We  report
statistics about the ground truth for reference, but we do not use this
information during synthesis. For instance, we do not assume knowledge of  
which variables a missing expression depends on.


Adapting verification benchmarks for synthesis by sketching requires a number of
steps, some of which are non-trivial. We discuss the most salient points of
these steps next.

\subsubsection*{Holes}

For each protocol we performed many synthesis experiments by varying the number
of holes in the protocol sketch. All our experiments, as well as instructions
for reproducing them, can be found on GitHub~\cite{scythe-full-results}.
Representative experiments are summarized in \Table~\ref{tab:results}, explained
below.

\subsubsection*{Grammars}

Each hole requires a grammar. \tool is flexible; the user can provide a
different grammar for each hole, or reuse grammars across holes. \tool grammars
are {\em modular} in the sense that they contain a {\em general-purpose} part
(e.g., the grammar of boolean or arithmetic or set expressions) plus a {\em
hole-specific} part (e.g., the terminals which are the hole's arguments). We
implemented a library that allows to build  grammars by (1) automatically
constructing non-terminals based on the types of the hole's arguments and (2)
exposing to the user common sub-grammars that can be deployed across protocols.

\subsubsection*{Liveness and Fairness}

Our benchmarks come from existing suites focusing on {\em safety}
verification~\cite{DBLP:conf/fmcad/SchultzDT22,ivybench}. Performing
synthesis against only safety properties often results in {\em \spurious}
solutions that satisfy safety in trivial ways  
(e.g. by filling a pre-hole with the expression \False). Therefore, we augment
each benchmark with additional liveness properties and any necessary fairness
constraints.

\subsubsection*{Implementability Constraints}

In addition to excluding \spurious solutions by adding extra properties, we
sometimes need to exclude {\em unimplementable} solutions, for instance,
solutions violating implicit communication/observability constraints between the
protocol processes. For example, replacing the post-condition in
line~\ref{line:vote-yes-vote-yes} of \Fig~\ref{fig:protocol_example} with
$\textit{vote\_yes}' = \emptyset$ results in an unimplementable protocol because
a node cannot directly change the vote state of another node. To avoid such
solutions, we used arrays instead of sets (e.g., \textit{vote\_yes} is an array
mapping process ids to booleans). Then, we restricted the grammar to only
contain array access expressions with appropriate indices.

\subsubsection*{Explicitly Modeling the Environment}

We had to modify several of the verification benchmarks
of~\cite{DBLP:conf/fmcad/SchultzDT22} in order to explicitly separate  the
(controllable) protocol from its (uncontrollable) environment, so as to prevent
synthesis of parts belonging to the environment.


\subsubsection*{Results}

\input{figs/results_table}

\TLA LOC is the number of lines of code of the ground truth \TLA protocol
specification, which is the same as the lines of code in the sketch and the
synthesized protocols, since all synthesized expressions are printed to one
line, regardless of size. ID refers to the number used to identify the
experiment in the full results table~\cite{scythe-full-results}. ``\#pre/post
holes'' is the number of pre- and post-holes in the sketch, and $k$ refers to $k
= k_1 + k_2 + ... + k_n$, where each $k_i$ is the {\em size} of the expression
(c.f. \Section~\ref{sec_EG}) used in the ground truth protocol for the $i$th
hole. ``gram. LOC'' is the number of lines (non-boilerplate) code in the python
script used to generate the grammar. Every protocol uses the same grammar
generation script, regardless of which or how many holes are poked.

We report Execution Statistics for the tool with and without equivalence
reduction. The column ``generated / model checked'' reports the number of
completions generated by the tool vs those model checked (the rest were pruned).
The column $k'$ is either the size of the expression found by the tool, or the
size of the largest expression the tool considered before it timed out (marked
with a $\geq$ symbol). If there are multiple holes then $k'$ is the sum of all
expression sizes. Column ``total / model checking time'' reports the total
execution time vs the time devoted to model checking (both in seconds). TO
indicates that the tool timed out after 1 hour; TO\textsuperscript{**} is
explained below. Note that TLC is called without a timeout; hence, it performs
exhaustive model checking on the finite protocol instances  specified by the
user configuration. 

As \Table~\ref{tab:results} shows, our efficient expression generation technique
with equivalence reduction achieves impressive results, sometimes reducing the
number of generated expressions by more than three orders of magnitude (c.f.
\raft ID 121 where only 271 expressions are generated with reduction, vs $>$
690,000 without reduction at the TO point). In all cases, the number of
completions model checked is much smaller than those generated, which shows how
critical pruning constraints are to scalability. With equivalence reduction,
execution time is typically dominated by model checking, although there are
exceptions (e.g. 2pc). Without equivalence reduction, the time is typically
dominated by expression generation, which demonstrates the importance of the
equivalence reduction.

Qualitatively, \tool often synthesizes large, non-trivial expressions, e.g.,
single expressions of size 14 in the cases of \raft ID 121 and
\raft-big ID 714, and multiple expressions of combined size up to 18 in
other cases.  Expressions~(\ref{expressionintro})
and~(\ref{expressionintro2}) shown in~\Section~\ref{sec_intro} are two concrete
examples of synthesized expressions.


\subsubsection*{Novel Solutions}

The protocols synthesized by \tool were often identical (or almost identical, up
to commutativity of an operator such as $\land$, etc.) to the ground truth. In
other cases, however, \tool found novel, non-\spurious solutions. \tool often
found solutions with shorter expressions. One notable example comes from the
experiment 2pc ID 303, where instead of the ground-truth expression $e_1 :=
\emptyset \neq (P \setminus A) \cup (P \cap N)$, \tool found the expression $e_2
:= P \neq (A \setminus N)$, where $P$ is the set of all nodes, $A$ is the set of
alive nodes, and $N$ is the set of nodes that voted no. So $e_1$ says ``There is
a node that is dead or there is at least one node that voted no.'' In the
context of the protocol, $e_1$ and $e_2$ are equivalent, but a proof requires
the subtle reasoning that both $A$ and $N$ are subsets of $P$, since $P$ is the
set of all nodes.

\subsubsection*{Correctness of Infinite Instances}

A protocol synthesized by \tool is a solution to Problem~\ref{problem:instance},
i.e., is correct for the finite instance specified by the user. This correctness
follows from the fact that during the synthesis loop the verifier (TLC) exhausts
the state space of the specified finite instance. As it turns out, the protocols
of \Table~\ref{tab:results} produced by \tool are also solutions to
Problem~\ref{problem:protocol}. Specifically, for each protocol of
\Table~\ref{tab:results} except peterson, we used the \TLA Proof System
(TLAPS)~\cite{cousineau2012tlaps} to prove that all instances of that protocol
satisfy the key safety property (our TLAPS proofs did not consider liveness; the
four peterson variants involve only two processes and need no extra
verification). For our TLAPS proofs we used techniques similar to those reported
in~\cite{SchultzDardikTripakisCPP2022}.

In all but two cases, the initial solutions produced by \tool proved to be
correct. For \raft ID 343 and \raft-big ID 709, \tool initially produced a
solution which is correct for up to 3 nodes, but which we were surprised to find
is incorrect for 4 or more nodes. To address this scenario, we added to the tool
an {\em extra-check} option to perform an additional model checking step with
larger parameter values than those used in the synthesis loop, right before
outputting the final solution (if the extra-check fails, the tool continues to
search for a solution). \tool with extra-check found a correct (for all
instances) solution for \raft ID 343 in 750 secs (this includes the time spent
for extra-checks). \tool also found a correct (for all instances) solution for
\raft-big ID 709, although it timed out after a total of four hours
(TO\textsuperscript{**}) while performing the final extra-check for 4
nodes---that single final extra-check took about 2 hours.

%% file: figs/results_table.tex
\begin{table*}[]
\centering
\resizebox{0.8\textwidth}{!}{%
\begin{tabular}{|l|l|l|l|l|l|l|r|l|l|r|l|}
\hline
\multicolumn{2}{|l}{}                                         & \multicolumn{4}{|l}{}                                                                                                       & \multicolumn{6}{|c|}{Execution Stats}                                                                                                                                                                                                                                                                \\
\multicolumn{2}{|l}{}                                         & \multicolumn{4}{|c|}{Sketch Parameters}                                                                                     & \multicolumn{3}{c}{w/o eq. reduction}                                                                                                            & \multicolumn{3}{c|}{w/ eq. reduction}                                                                                                                                       \\
\hline
Protocol & \begin{tabular}[c]{@{}l@{}}\TLA\\ LOC\end{tabular} & ID & \begin{tabular}[c]{@{}l@{}}\#pre/post\\ holes\end{tabular} & $k$ & \begin{tabular}[c]{@{}l@{}}gram.\\ LOC\end{tabular} & \begin{tabular}[c]{@{}l@{}}\textbf{generated}/\\ model checked\end{tabular} & \multicolumn{1}{|c|}{$k'$} & \begin{tabular}[c]{@{}l@{}}\textbf{total} / model\\ checking time\end{tabular} & \begin{tabular}[c]{@{}l@{}}\textbf{generated}/\\ model checked\end{tabular} & $k'$ & \begin{tabular}[c]{@{}l@{}}\textbf{total} / model\\ checking time\end{tabular} \\
\hline
decentr. lock & 48 & 486 & 1 / 3 & 21 & 28 & \textbf{93403} / 136 & 16 & \textbf{674} / 193 & \textbf{3020} / 117 & \textit{16} & \textbf{190} / 175 \\
\hline
lock\_serv & 83 & 599 & 2 / 6 & 16 & 22 & \textbf{384569} / 85 & 16 & \textbf{2116} / 134 & \textbf{7483} / 80 & \textit{16} & \textbf{159} / 120 \\
lock\_serv & 83 & 611 & 2 / 6 & 16 & 22 & \textbf{665463} / 104 & $\geq$16 & \textbf{TO} / 1094 & \textbf{4064} / 84 & \textit{16} & \textbf{145} / 124 \\
\hline
peterson & 105 & 475 & 3 / 1 & 19 & 53 & \textbf{442553} / 324 & 12 & \textbf{2731} / 453 & \textbf{485} / 243 & \textit{12} & \textbf{348} / 346 \\
peterson & 105 & 375 & 2 / 2 & 19 & 53 & \textbf{583201} / 264 & 13 & \textbf{3355} / 356 & \textbf{1369} / 267 & \textit{13} & \textbf{364} / 357 \\
peterson & 105 & 413 & 3 / 1 & 22 & 53 & \textbf{582616} / 353 & $\geq$12 & \textbf{TO} / 602 & \textbf{7073} / 1259 & \textit{15} & \textbf{1809} / 1753 \\
peterson & 105 & 547 & 2 / 6 & 22 & 53 & \textbf{643529} / 167 & $\geq$16 & \textbf{TO} / 483 & \textbf{5569} / 222 & \textit{17} & \textbf{329} / 301 \\
\hline
2pc & 134 & 303 & 2 / 0 & 15 & 46 & \textbf{690411} / 24 & $\geq$8 & \textbf{TO} / 1494 & \textbf{65994} / 23 & \textit{9} & \textbf{388} / 43 \\
2pc & 134 & 558 & 3 / 5 & 18 & 46 & \textbf{410675} / 87 & 14 & \textbf{2301} / 179 & \textbf{89027} / 88 & \textit{14} & \textbf{629} / 171 \\
2pc & 134 & 485 & 2 / 2 & 17 & 46 & \textbf{681190} / 44 & $\geq$10 & \textbf{TO} / 492 & \textbf{493492} / 55 & \textit{11} & \textbf{2654} / 110 \\
2pc & 134 & 513 & 2 / 6 & 18 & 46 & \textbf{642178} / 162 & $\geq$18 & \textbf{TO} / 1641 & \textbf{98009} / 211 & \textit{18} & \textbf{886} / 382 \\
\hline
consensus & 127 & 624 & 2 / 6 & 17 & 56 & \textbf{550501} / 97 & 17 & \textbf{3442} / 606 & \textbf{9988} / 63 & \textit{17} & \textbf{427} / 375 \\
consensus & 127 & 550 & 2 / 6 & 22 & 56 & \textbf{483126} / 162 & $\geq$17 & \textbf{TO} / 1116 & \textbf{53994} / 286 & \textit{18} & \textbf{2291} / 2011 \\
\hline
sharded\_kv & 112 & 302 & 1 / 1 & 13 & 44 & \textbf{248298} / 13 & 13 & \textbf{1325} / 49 & \textbf{469} / 14 & \textit{13} & \textbf{59} / 57 \\
sharded\_kv & 112 & 365 & 1 / 3 & 22 & 44 & \textbf{611512} / 129 & $\geq$16 & \textbf{TO} / 941 & \textbf{3149} / 149 & \textit{17} & \textbf{472} / 455 \\
\hline
\raft & 174 & 463 & 1 / 3 & 21 & 82 & \textbf{64832} / 64 & 14 & \textbf{462} / 128 & \textbf{1958} / 65 & \textit{14} & \textbf{139} / 129 \\
\raft & 174 & 343 & 2 / 0 & 21 & 82 & \textbf{608586} / 215 & $\geq$11 & \textbf{TO} / 484 & \textbf{41411} / 251 & \textit{17} & \textbf{750} / 530 \\
\raft & 174 & 121 & 1 / 0 & 18 & 82 & \textbf{694552} / 12 & $\geq$12 & \textbf{TO} / 3589 & \textbf{271} / 13 & \textit{14} & \textbf{31} / 27 \\
\hline
\raft-big & 304 & 708 & 1 / 3 & 21 & 85 & \textbf{67237} / 78 & 14 & \textbf{1815} / 1465 & \textbf{3155} / 84 & \textit{14} & \textbf{1668} / 1651 \\
\raft-big & 304 & 709 & 2 / 0 & 21 & 85 & \textbf{2231397} / 220 & $\geq$12 & \textbf{TO\textsuperscript{**}} / 3950 & \textbf{282106} / 252 & \textit{$\geq$17} & \textbf{TO}\textsuperscript{**} / 12943 \\
\raft-big & 304 & 710 & 1 / 0 & 18 & 85 & \textbf{658492} / 12 & $\geq$12 & \textbf{TO} / 3588 & \textbf{1369} / 13 & \textit{14} & \textbf{368} / 359 \\
\raft-big & 304 & 714 & 1 / 7 & 25 & 85 & \textbf{221648} / 101 & 18 & \textbf{2662} / 1519 & \textbf{6500} / 102 & \textit{18} & \textbf{1802} / 1768 \\
\hline
\end{tabular}
}
\caption{}
\label{tab:results}
\vspace{-3em}
\end{table*}

%% file: sections/related_work.tex
\section{Related Work}
\label{sec:rel-work}

Past works synthesize explicit-state, finite-state
machines~\cite{ScenariosHVC2014,DBLP:journals/sigact/AlurT17,DBLP:conf/atva/EgolfT23,DBLP:journals/sttt/FinkbeinerS13}.
In contrast, we synthesize symbolic and parameterized infinite-state machines.
TRANSIT~\cite{DBLP:conf/pldi/UdupaRDMMA13} cannot process counterexamples
automatically and requires a human in the synthesis loop.
\cite{DBLP:journals/acta/MirzaieFJB20} and \cite{DBLP:conf/cav/BloemBJ16} use
{\em cut-off} techniques which only apply to a special class of self-stabilizing
protocols in symmetric networks, and \cite{DBLP:conf/tacas/JaberWJKS23} study a
special class of distributed agreement-based systems.
\cite{DBLP:conf/opodis/Lazic0WB17} consider only threshold-guarded distributed
protocols. In contrast, our work applies to general distributed protocols. 

As discussed in \Section~\ref{sec_intro}, \cite{DBLP:conf/cav/AlurRSTU15}
synthesize interpretations of uninterpreted functions represented as finite
lookup tables, whereas we synthesize symbolic expressions directly. We use \TLA
models with parameterized actions. In contrast, \cite{DBLP:conf/cav/AlurRSTU15}
use extended finite state machines (EFSMs) which do not have parameterized
actions. It is unclear whether expressions such as~(\ref{expressionintro})
and~(\ref{expressionintro2}) on page~\pageref{expressionintro2} could be
synthesized by~\cite{DBLP:conf/cav/AlurRSTU15}.

Like~\cite{DBLP:conf/cav/AlurRSTU15}, we use CEGIS and our counterexample
encodings are similar. Unlike~\cite{DBLP:conf/cav/AlurRSTU15}, we rely neither
on an external SyGuS solver nor on an SMT solver.
\cite{DBLP:conf/cav/AlurRSTU15} encodes the search space of candidate
interpretations as SMT formulas and calls an SMT solver to generate the next
candidate. SMT queries are both expensive and numerous in the context of CEGIS.
In contrast, we use efficient grammar enumeration techniques and we bypass SMT
solvers by checking candidate expressions directly against the pruning
constraints (\Section~\ref{sec:algorithm}).

Like~\cite{DBLP:conf/cav/AlurRSTU15}, our tool synthesizes solutions that are
guaranteed correct only up to the finite instances model checked in the CEGIS
loop. Unlike~\cite{DBLP:conf/cav/AlurRSTU15}, we went one step further and
proved with TLAPS that the solutions produced by our tool are actually correct
for all instances. As discussed in \Section~\ref{sec:evaluation}, this step is
not redundant: there were surprising cases of solutions which are correct for 3
nodes but not for 4 or more nodes. It is unclear whether the protocols
synthesized in~\cite{DBLP:conf/cav/AlurRSTU15} are correct beyond the finite
model checked instances.

\cite{KatzPeled2009} use genetic programming and \cite{ml-based-synthesis} use
machine learning for synthesis. Generally, these approaches are not guaranteed
to find a solution even if one exists, i.e. they are incomplete. In contrast,
our approach is complete.

None of the works cited above use syntax to guide the search, none use
equivalence of expressions with short-circuiting to reduce the search space, and
none handle state variables with infinite domains. To our knowledge, ours is the
only truly syntax-guided synthesis method for symbolic, parameterized
distributed protocols.

Existing SyGuS solvers use SMT formulas to express properties, and are therefore
not directly applicable to distributed protocol synthesis which requires
temporal logic properties. But our techniques for generating expressions and
checking them against pruning constraints are generally related to term
enumeration strategies used in SyGuS~\cite{DBLP:conf/fmcad/AlurBJMRSSSTU13}.
Both EUSolver~\cite{DBLP:conf/tacas/AlurRU17} and cvc4sy~\cite{RBN19} are SyGuS
solvers that generate larger expressions from smaller expressions. EUSolver uses
divide-and-conquer techniques in combination with decision tree learning and is
quite different from our approach. To our knowledge, EUSolver does not employ
equivalence reduction. The ``fast term enumeration strategy'' of cvc4sy is
similar to our cache-based approach and also uses equivalence reduction
techniques. To our knowledge, cvc4sy does not use short-circuiting.


%% file: sections/conclusion.tex
\section{Conclusion}

We present the only, to our knowledge, truly syntax-guided synthesis method for
symbolic, parameterized, infinite-state distributed protocols. We show
experimentally that our method and tool are able to synthesize non-trivial
completions across a broad set of non-trivial protocols written in \TLA, and
prove that these completions generalize correctly (i.e., preserve safety) in all
possible instances.

Future work includes:
(1) further ways to reduce the search space and short-circuit parts of the
search;
(2) optimization of the \tool-TLC interface to avoid running a new instance of
(and repeatedly initializing) TLC each time \tool needs to check a candidate
protocol;
and
(3) automating the final, all-instances verification step (generally an
undecidable problem), by potentially combining TLAPS with state of the art
inductive invariant inference techniques~\cite{schultz2024scalablearxiv}.

%% file: sections/theorems.tex
\subsection{Proof of Theorem~\ref{thm_all}}
\label{sec:proof-theorem-all}

Theorem~\ref{thm_all} follows from the four theorems below.

\input{sections/theoremsafety}

\input{sections/theoremdeadlock}

\input{sections/theoremliveness}

\input{sections/theoremstuttering}

%% file: sections/theoremsafety.tex
\begin{theorem}
    Let \run be a safety violation of a completion of the sketch $S$. Then
    $\pi_\textit{safe}(\run)$ is optimal w.r.t. \run and $S$.
    \label{thm:safe}
\end{theorem}
\begin{proof}
    For brevity, $\pi := \pi_\textit{safe}(\run)$. To show that $\pi$ is
    optimal, we must show that for any completion \prot of the sketch $S$, \prot
    satisfies $\pi$ if and only if \run is not run of \prot. First suppose \prot
    satisfies $\pi$. Then $\pi$ is not an empty disjunction and moreover there
    exists a $\tau_{A(\vec v),i}$ in $\pi$ that is satisfied by \prot. This
    $\tau$-term corresponds to some transition $s\xrightarrow{A(\vec v)}t$ in
    \run. If the $h_i$ corresponding to the $\tau$-term is a pre-hole, then the
    action $A(\vec v)$ is disabled in state $s$ of \prot. Suppose $h_i$ is a
    post-hole corresponding to the state variable $x$. Then $x$ has some value
    $v_x$ in $t$. Because \prot satisfies $\pi$, we know that after taking
    action $A(\vec v)$ in state $s$, $x$ has some value $v_x^\star \neq v_x$. In
    either case, \run is not a run of \prot because it cannot transition from
    $s$ to $t$ by action $A(\vec v)$.

    Now suppose \prot does not satisfy $\pi$. Then none of the $\tau$-terms in
    $\pi$ are satisfied. Let $T = s\xrightarrow{A(\vec v)}t$ be a transition in
    \run. We will show that \prot contains all such $T$ and therefore \run is a
    run of \prot. There are two cases: (1) the action of $T$ has holes in it, or
    (2) it does not. In case (2), $T$ is a transition that is present in every
    completion of the sketch. In case (1), we can leverage that \prot violates
    all $\tau_{A(\vec v),i}$ that were constructed for $T$. If it violates all
    $\tau_{A(\vec v),i}$ for all pre-holes, then $A(\vec v)$ is enabled in state
    $s$. If it violates all $\tau_{A(\vec v),i}$ for all post-holes, then $t$ is
    a successor of $s$ by action $A(\vec v)$ in \prot. 
\end{proof}

%% file: sections/theoremdeadlock.tex
\begin{theorem}
    Let \run be a deadlock violation of a completion of the sketch $S$.
    Then $\pi_\textit{dead}(\run)$ is sub-optimal w.r.t. \run and $S$.
    \label{thm:dead}
\end{theorem}
\begin{proof}
    Let $\pi := \pi_\textit{dead}(\run)$ for brevity. $\pi$ is under-pruning
    because although $\rho$ terms ensure that some pre-condition for some action
    $A$ is weakened for deadlocked state $s_k$, it is possible that multiple
    pre-conditions need to be weakened in order for $A$ to be taken in $s_k$.

    Let $\prot$ be a completion of the sketch $S$ that does not satisfy $\pi$.
    To show that $\pi$ is not over-pruning, we must show that \run is a run of
    \prot and that the final state is deadlocked in \prot. As with the
    $\pi_\textit{safe}$ proof, we know that $\prot$ has all the transitions in
    \run, since all $\tau$-terms are violated. Furthermore, we know that the
    final state of \run is deadlocked in $\prot$ because all $\rho$-terms are
    violated and therefore no pre-condition is weak enough to be taken in order
    to escape the deadlocked state. 
\end{proof}

%% file: sections/theoremliveness.tex
\begin{theorem}
    Let \run be a liveness violation of a completion of the sketch $S$.
    Then $\pi_\textit{live}(\run)$ is sub-optimal w.r.t. \run and $S$.
    \label{thm:live}
\end{theorem}
\begin{proof}
    Let $\pi := \pi_\textit{live}(\run)$. As with the deadlock case, $\pi$ is
    under-pruning because \prot satifying $\pi$ may only weaken one
    pre-condition of a fair action where it is necessary to weaken multiple
    pre-conditions to enable a fair action and make a cycle unfair.

    Let $\prot$ be a completion of the sketch $S$ that does not satisfy $\pi$.
    Then all $\tau$-terms are violated, so \run is a run of \prot, so long as
    fairness constraints are satisfied. Fairness constraints are satisfied
    because all of the $\rho$-terms in $\pi$ are violated. I.e., $\rho$-terms
    ensure there does not exist a fair action in \prot that is enabled in the
    cycle of $\run$.
\end{proof}

%% file: sections/theoremstuttering.tex
\begin{theorem}
    Let \run be a stuttering violation of a completion of the sketch $S$. Then
    $\pi_\textit{stut}(\run)$ is sub-optimal w.r.t. \run and $S$.
    \label{thm:stut}
\end{theorem}
\begin{proof}
    Let $\pi := \pi_\textit{stut}(\run)$ and suppose $\prot_0$ is a completion
    of $S$ that exhibits \run. As with the deadlock and liveness violations,
    $\pi$ is under-pruning because \prot satisfying $\pi$ may only weaken one
    pre-condition of a fair action where it is necessary to weaken multiple
    pre-conditions to enable a fair action and make stuttering unfair.

    Let $\prot$ be a completion of the sketch $S$ that does not satisfy $\pi$.
    We must show that \run is a run of \prot. All terms of
    $\pi_\textit{safe}(\run)$ are violated, so the last state, $s_k$, of \run is
    reachable in \prot by taking the sequence of transitions in $r$. Now, each
    fair action $A$ of \prot is either enabled or disabled in $s_k$. We must
    show that all enabled fair actions are self-looping. Because the terms in
    $\pi_\rho'$ are violated, we know that the non-self-looping fair actions
    that were disabled in state $s_k$ of $\prot_0$ are also disabled in state
    $s_k$ of $\prot$. Because the terms in $\pi_\tau$ are violated, we know that
    states that were self-looping in $\prot_0$ are still self-looping in
    $\prot$, if they are enabled in \prot.
\end{proof}

%% file: main.bbl
\begin{thebibliography}{10}

\bibitem{DBLP:conf/fmcad/AlurBJMRSSSTU13}
Rajeev Alur, Rastislav Bod{\'{\i}}k, Garvit Juniwal, Milo M.~K. Martin, Mukund
  Raghothaman, Sanjit~A. Seshia, Rishabh Singh, Armando Solar{-}Lezama, Emina
  Torlak, and Abhishek Udupa.
\newblock Syntax-guided synthesis.
\newblock In {\em Formal Methods in Computer-Aided Design, {FMCAD} 2013,
  Portland, OR, USA, October 20-23, 2013}, pages 1--8. {IEEE}, 2013.

\bibitem{ScenariosHVC2014}
Rajeev Alur, Milo Martin, Mukund Raghothaman, Christos Stergiou, Stavros
  Tripakis, and Abhishek Udupa.
\newblock {Synthesizing Finite-state Protocols from Scenarios and
  Requirements}.
\newblock In {\em Haifa Verification Conference}, volume 8855 of {\em LNCS}.
  Springer, 2014.

\bibitem{DBLP:conf/tacas/AlurRU17}
Rajeev Alur, Arjun Radhakrishna, and Abhishek Udupa.
\newblock Scaling enumerative program synthesis via divide and conquer.
\newblock In {\em Tools and Algorithms for the Construction and Analysis of
  Systems - 23rd International Conference, {TACAS} 2017}, volume 10205 of {\em
  Lecture Notes in Computer Science}, pages 319--336, 2017.

\bibitem{DBLP:conf/cav/AlurRSTU15}
Rajeev Alur, Mukund Raghothaman, Christos Stergiou, Stavros Tripakis, and
  Abhishek Udupa.
\newblock Automatic completion of distributed protocols with symmetry.
\newblock In Daniel Kroening and Corina~S. Pasareanu, editors, {\em Computer
  Aided Verification - 27th International Conference, {CAV}}, volume 9207 of
  {\em Lecture Notes in Computer Science}, pages 395--412. Springer, 2015.

\bibitem{DBLP:journals/sigact/AlurT17}
Rajeev Alur and Stavros Tripakis.
\newblock Automatic synthesis of distributed protocols.
\newblock {\em {SIGACT} News}, 48(1):55--90, 2017.

\bibitem{DBLP:conf/cav/BloemBJ16}
Roderick Bloem, Nicolas Braud{-}Santoni, and Swen Jacobs.
\newblock Synthesis of self-stabilising and byzantine-resilient distributed
  systems.
\newblock In Swarat Chaudhuri and Azadeh Farzan, editors, {\em Computer Aided
  Verification - 28th International Conference, {CAV}}, volume 9779 of {\em
  Lecture Notes in Computer Science}, pages 157--176. Springer, 2016.

\bibitem{Buchman2016TendermintBF}
Ethan Buchman.
\newblock Tendermint: Byzantine fault tolerance in the age of blockchains.
\newblock 2016.

\bibitem{Buterin2013}
Vitalik Buterin.
\newblock Ethereum white paper: A next generation smart contract \&
  decentralized application platform.
\newblock 2013.

\bibitem{corbett2013spanner}
James~C Corbett, Jeffrey Dean, Michael Epstein, Andrew Fikes, Christopher
  Frost, Jeffrey~John Furman, Sanjay Ghemawat, Andrey Gubarev, Christopher
  Heiser, Peter Hochschild, et~al.
\newblock Spanner: Google’s globally distributed database.
\newblock {\em ACM Transactions on Computer Systems (TOCS)}, 31(3):1--22, 2013.

\bibitem{cousineau2012tlaps}
Denis Cousineau, Damien Doligez, Leslie Lamport, Stephan Merz, Daniel Ricketts,
  and Hernan Vanzetto.
\newblock {TLA+ Proofs}.
\newblock {\em 18th International Symposium on Formal Methods (FM 2012)},
  7436:147--154, January 2012.

\bibitem{decandia2007dynamo}
Giuseppe DeCandia, Deniz Hastorun, Madan Jampani, Gunavardhan Kakulapati,
  Avinash Lakshman, Alex Pilchin, Swaminathan Sivasubramanian, Peter Vosshall,
  and Werner Vogels.
\newblock Dynamo: Amazon's highly available key-value store.
\newblock {\em ACM SIGOPS operating systems review}, 41(6):205--220, 2007.

\bibitem{scythe-full-results}
Derek Egolf.
\newblock scythe-fmcad2024.
\newblock \url{https://github.com/egolf-cs/scythe-fmcad2024}.

\bibitem{DBLP:conf/atva/EgolfT23}
Derek Egolf and Stavros Tripakis.
\newblock Synthesis of distributed protocols by enumeration modulo
  isomorphisms.
\newblock In {\em {ATVA} 2023 - Part {I}}, Lecture Notes in Computer Science,
  pages 270--291. Springer, 2023.

\bibitem{DBLP:journals/corr/Engelfriet15}
Joost Engelfriet.
\newblock Tree automata and tree grammars.
\newblock {\em CoRR}, abs/1510.02036, 2015.

\bibitem{DBLP:journals/sttt/FinkbeinerS13}
Bernd Finkbeiner and Sven Schewe.
\newblock Bounded synthesis.
\newblock {\em Int. J. Softw. Tools Technol. Transf.}, 15(5-6):519--539, 2013.

\bibitem{ivybench}
Aman Goel.
\newblock {IvyBench}.
\newblock \url{https://github.com/aman-goel/ivybench}, Accessed: 2024-04-22.

\bibitem{2021ic3posymmetry}
Aman Goel and Karem Sakallah.
\newblock {On Symmetry and Quantification: A New Approach to Verify Distributed
  Protocols}.
\newblock In {\em NASA Formal Methods: 13th International Symposium, NFM 2021},
  page 131–150, 2021.

\bibitem{GoelSakallahFMCAD2021}
Aman Goel and Karem~A. Sakallah.
\newblock Towards an automatic proof of lamport’s paxos.
\newblock In {\em 2021 Formal Methods in Computer Aided Design (FMCAD)}, pages
  112--122, 2021.

\bibitem{GulwaniPolozovSingh2017}
Sumit Gulwani, Oleksandr Polozov, and Rishabh Singh.
\newblock Program synthesis.
\newblock {\em Foundations and Trends in Programming Languages}, 4(1-2):1--119,
  2017.

\bibitem{2021swisshance}
Travis Hance, Marijn Heule, Ruben Martins, and Bryan Parno.
\newblock {Finding Invariants of Distributed Systems: It{\textquoteright}s a
  Small (Enough) World After All}.
\newblock In {\em 18th USENIX Symposium on Networked Systems Design and
  Implementation (NSDI 21)}, pages 115--131. USENIX Association, April 2021.

\bibitem{ml-based-synthesis}
Yujie Hui, Drew Ripberger, Xiaoyi Lu, and Yang Wang.
\newblock Learning distributed protocols with zero knowledge.
\newblock In {\em Machine Learning for Systems at NeurIPS 2023}, 2023.

\bibitem{DBLP:conf/tacas/JaberWJKS23}
Nouraldin Jaber, Christopher Wagner, Swen Jacobs, Milind Kulkarni, and Roopsha
  Samanta.
\newblock Synthesis of distributed agreement-based systems with
  efficiently-decidable verification.
\newblock In {\em {TACAS} 2023}, volume 13994 of {\em Lecture Notes in Computer
  Science}, pages 289--308. Springer, 2023.

\bibitem{DBLP:journals/corr/JacobsB14}
Swen Jacobs and Roderick Bloem.
\newblock Parameterized synthesis.
\newblock {\em Log. Methods Comput. Sci.}, 10(1), 2014.

\bibitem{KatzPeled2009}
Gal Katz and Doron Peled.
\newblock Synthesizing solutions to the leader election problem using model
  checking and genetic programming.
\newblock In {\em Haifa Verification Conference}, HVC'09, page 117–132.
  Springer, 2009.

\bibitem{lamport2002specifying}
Leslie Lamport.
\newblock {\em {Specifying Systems: The TLA+ Language and Tools for Hardware
  and Software Engineers}}.
\newblock Addison-Wesley, Jun 2002.

\bibitem{DBLP:conf/opodis/Lazic0WB17}
Marijana Lazic, Igor Konnov, Josef Widder, and Roderick Bloem.
\newblock Synthesis of distributed algorithms with parameterized threshold
  guards.
\newblock In {\em 21st International Conference on Principles of Distributed
  Systems, {OPODIS}}, volume~95 of {\em LIPIcs}, pages 32:1--32:20. Schloss
  Dagstuhl - Leibniz-Zentrum f{\"{u}}r Informatik, 2017.

\bibitem{DBLP:journals/acta/MirzaieFJB20}
Nahal Mirzaie, Fathiyeh Faghih, Swen Jacobs, and Borzoo Bonakdarpour.
\newblock Parameterized synthesis of self-stabilizing protocols in symmetric
  networks.
\newblock {\em Acta Informatica}, 57(1-2):271--304, 2020.

\bibitem{NewcombeAmazon2015}
Chris Newcombe, Tim Rath, Fan Zhang, Bogdan Munteanu, Marc Brooker, and Michael
  Deardeuff.
\newblock {How Amazon Web Services Uses Formal Methods}.
\newblock {\em Commun. ACM}, 58(4):66--73, March 2015.

\bibitem{ongaro2014search}
Diego Ongaro and John Ousterhout.
\newblock In search of an understandable consensus algorithm.
\newblock In {\em 2014 {USENIX} Annual Technical Conference ({USENIX} {ATC}
  14)}, pages 305--319. {USENIX} Association, June 2014.

\bibitem{padonpaxosEPR}
Oded Padon, Giuliano Losa, Mooly Sagiv, and Sharon Shoham.
\newblock {Paxos Made EPR: Decidable Reasoning about Distributed Protocols}.
\newblock {\em Proc. ACM Program. Lang.}, 1(OOPSLA), Oct 2017.

\bibitem{2016padonivy}
Oded Padon, Kenneth~L. McMillan, Aurojit Panda, Mooly Sagiv, and Sharon Shoham.
\newblock {Ivy: Safety Verification by Interactive Generalization}.
\newblock In {\em Proceedings of the 37th ACM SIGPLAN Conference on Programming
  Language Design and Implementation}, PLDI '16, pages 614--630. ACM, 2016.

\bibitem{PnueliRosner90}
A.~Pnueli and R.~Rosner.
\newblock Distributed reactive systems are hard to synthesize.
\newblock In {\em Proceedings of the 31th IEEE Symposium on Foundations of
  Computer Science}, pages 746--757, 1990.

\bibitem{RBN19}
Andrew Reynolds, Haniel Barbosa, Andres N{\"o}tzli, Cesare Tinelli, and Clark
  Barrett.
\newblock {CVC4SY}: Smart and fast term enumeration for syntax-guided
  synthesis.
\newblock In Isil Dillig and Serdar Tasiran, editors, {\em Proceedings of the
  31st International Conference on Computer Aided Verification (CAV)}, volume
  11561 of {\em Lecture Notes in Computer Science}, pages 74--83. Springer,
  July 2019.

\bibitem{schultz2024scalablearxiv}
William Schultz, Edward Ashton, Heidi Howard, and Stavros Tripakis.
\newblock {Scalable, Interpretable Distributed Protocol Verification by
  Inductive Proof Slicing}.
\newblock arXiv eprint 2404.18048, 2024.

\bibitem{SchultzDardikTripakisCPP2022}
William Schultz, Ian Dardik, and Stavros Tripakis.
\newblock Formal verification of a distributed dynamic reconfiguration
  protocol.
\newblock In {\em Proceedings of the 11th ACM SIGPLAN International Conference
  on Certified Programs and Proofs}, CPP 2022, page 143–152. ACM, 2022.

\bibitem{DBLP:conf/fmcad/SchultzDT22}
William Schultz, Ian Dardik, and Stavros Tripakis.
\newblock {Plain and Simple Inductive Invariant Inference for Distributed
  Protocols in {TLA}\({}^{\mbox{+}}\)}.
\newblock In {\em 22nd Formal Methods in Computer-Aided Design, {FMCAD} 2022},
  pages 273--283. {IEEE}, 2022.

\bibitem{LezamaAPLAS2009}
Armando Solar-Lezama.
\newblock The sketching approach to program synthesis.
\newblock In {\em Proceedings of the 7th Asian Symposium on Programming
  Languages and Systems}, APLAS '09, pages 4--13. Springer, 2009.

\bibitem{ArmandoSTTT2013}
Armando Solar-Lezama.
\newblock Program sketching.
\newblock {\em Int. J. Softw. Tools Technol. Transf.}, 15(5-6):475--495, oct
  2013.

\bibitem{Thistle2005}
John~G. Thistle.
\newblock Undecidability in decentralized supervision.
\newblock {\em Systems \& Control Letters}, 54(5):503--509, 2005.

\bibitem{TripakisIPL}
Stavros Tripakis.
\newblock {Undecidable Problems of Decentralized Observation and Control on
  Regular Languages}.
\newblock {\em Information Processing Letters}, 90(1):21--28, April 2004.

\bibitem{DBLP:conf/pldi/UdupaRDMMA13}
Abhishek Udupa, Arun Raghavan, Jyotirmoy~V. Deshmukh, Sela Mador{-}Haim, Milo
  M.~K. Martin, and Rajeev Alur.
\newblock {TRANSIT:} specifying protocols with concolic snippets.
\newblock In Hans{-}Juergen Boehm and Cormac Flanagan, editors, {\em {ACM}
  {SIGPLAN} Conference on Programming Language Design and Implementation,
  {PLDI} '13, Seattle, WA, USA, June 16-19, 2013}, pages 287--296. {ACM}, 2013.

\bibitem{YaoTGN22}
Jianan Yao, Runzhou Tao, Ronghui Gu, and Jason Nieh.
\newblock {DuoAI: Fast, Automated Inference of Inductive Invariants for
  Verifying Distributed Protocols}.
\newblock In Marcos~K. Aguilera and Hakim Weatherspoon, editors, {\em 16th
  USENIX Symposium on Operating Systems Design and Implementation (OSDI 2022)},
  pages 485--501. {USENIX} Association, 2022.

\bibitem{YaoLivenessPOPL2024}
Jianan Yao, Runzhou Tao, Ronghui Gu, and Jason Nieh.
\newblock Mostly automated verification of liveness properties for distributed
  protocols with ranking functions.
\newblock {\em Proceedings of the ACM on Programming Languages (POPL)},
  8:1028--1059, jan 2024.

\bibitem{yao2021distai}
Jianan Yao, Runzhou Tao, Ronghui Gu, Jason Nieh, Suman Jana, and Gabriel Ryan.
\newblock {{DistAI}: {Data-Driven} Automated Invariant Learning for Distributed
  Protocols}.
\newblock In {\em 15th USENIX Symposium on Operating Systems Design and
  Implementation (OSDI 2021)}, pages 405--421. USENIX Association, July 2021.

\bibitem{tlcmodelchecker}
Yuan Yu, Panagiotis Manolios, and Leslie Lamport.
\newblock {Model Checking TLA+ Specifications}.
\newblock In Laurence Pierre and Thomas Kropf, editors, {\em Correct Hardware
  Design and Verification Methods}, pages 54--66, Berlin, Heidelberg, 1999.
  Springer Berlin Heidelberg.

\end{thebibliography}
